\documentclass[11pt]{article}

\usepackage{fullpage,amsthm,amssymb,amsmath,amsfonts,rotate}

\usepackage{graphicx,algorithmic,algorithm}
\usepackage[]{datetime}
\usepackage{aliascnt} 
\usepackage{hyperref,todonotes}
\usepackage{lmodern}
\usepackage{tikz}
\usetikzlibrary{calc,3d}
\usetikzlibrary{intersections,decorations.pathmorphing,shapes,decorations.pathreplacing,fit}

\usepackage{subcaption}
\usepackage{cleveref}

\usepackage[greek,english]{babel}
\usepackage[utf8]{inputenc}
\usepackage{alphabeta}

\usepackage{amsmath,ifthen,color,eurosym}
\usepackage{wrapfig,hyperref,cite}
\usepackage{latexsym,amsthm,amsmath,amssymb,url}

\extrafloats{100}

\usepackage{enumerate}
\usepackage{enumitem}

\hypersetup{
 pdftitle = {Edge degeneracy and admissibility: algorithmic and structural results
},
 pdfauthor=  {Stratis Limnios, Christophe Paul, Joanny Perret, Dimitrios Thilikos},
 colorlinks = true,
 urlcolor = black!50!red,
 linkcolor = black!50!red,
 citecolor = black!50!green,
 menucolor = {red}
 filecolor = {cyan},
 anchorcolor = {black}
 urlcolor = {magenta},
 runcolor = {cyan},
 linkbordercolor = {blue}
}

\newcommand{\remove}[1]{}
\usepackage{color}

\definecolor{Black}{rgb}{0,0, 0}
\definecolor{Blue}{rgb}{0, 0 ,1}
\definecolor{Red}{rgb}{1, 0 ,0}
\definecolor{White}{rgb}{1, 1, 1}
\definecolor{Grey}{rgb}{.6, .6, .6}
\definecolor{Mygreen}{rgb}{.0, .7, .0}
\definecolor{Yellow}{rgb}{.55,.55,0}
\definecolor{mustard}{rgb}{1.0, 0.86, 0.35}
\definecolor{applegreen}{rgb}{0.55, 0.71, 0.0}

\usepackage{tikz}
\usetikzlibrary{backgrounds}
\usetikzlibrary{decorations.pathreplacing}
\tikzset{black node/.style={draw, circle, fill = black, minimum size = 5pt, inner sep = 0pt}}
\tikzset{white node/.style={draw, circle, fill = white, minimum size =
  5pt, inner sep = 0pt}}
\tikzset{bag/.style={draw, circle, fill = white, minimum size = 0.75cm, inner sep = 0pt}}
\tikzset{smallbag/.style={draw, circle, fill = white, minimum size = 0.33cm, inner sep = 0pt, thin}}
\tikzset{normal/.style = {draw=none, fill = none, rectangle}}

\newcommand{\mynewtheorem}[2]{
 \newaliascnt{#1}{dummy}
 \newtheorem{#1}[#1]{#2}
 \aliascntresetthe{#1}
 \expandafter\def\csname #1autorefname\endcsname{#2}
}

\renewcommand{\deg}{{\bf deg}}

\newcommand{\eadm}{\delta_{\rm e}^{\infty}}

\theoremstyle{plain}
\mynewtheorem{theorem}{Theorem}
\mynewtheorem{proposition}{Proposition}
\mynewtheorem{corollary}{Corollary}
\mynewtheorem{lemma}{Lemma}

\theoremstyle{definition}
\mynewtheorem{definition}{Definition}
\theoremstyle{remark}
\mynewtheorem{sublemma}{Sublemma}
\mynewtheorem{claim}{Claim}
\mynewtheorem{remark}{Remark}
\mynewtheorem{conjecture}{Conjecture}
\mynewtheorem{question}{Question}       
\mynewtheorem{observation}{Observation}
\mynewtheorem{problem}{Problem}
\mynewtheorem{note}{Note}

\newcommand{\Nbb}{\mathbb{N}}

\newcommand{\intv}[1]{\left [ #1 \right ]}

\title{Edge Degeneracy: Algorithmic  and Structural Results}
\medskip

\author{Stratis Limnios\thanks{Data Science and Mining (DaSciM) team, Laboratoire d'Informatique (LIX)
École Polytechnique, Palaiseau, France.}~\thanks{Supported  by  project ESIGMA (ANR-17-CE23-0010).} \and Christophe Paul\thanks{LIRMM, Univ Montpellier, CNRS, Montpellier, France.}~$^{\dagger}$\thanks{Supported  by  project DEMOGRAPH (ANR-16-CE40-0028).} \and Joanny Perret\thanks{GIPSA-lab, Université de Grenoble, Grenoble, France.} \and Dimitrios M. Thilikos$^{\ddagger\dagger\S}$\thanks{Supported by the Research Council of Norway and the French Ministry of Europe and Foreign Affairs, via the Franco-Norwegian project PHC AURORA 2019.}}
\date{}

\begin{document}
 \remove{}  
 %
 %
 
 \maketitle


 \begin{abstract}
\noindent 
We consider a cops and robber game where 
the cops are blocking edges of a graph, while the robber occupies
its vertices. At each round of the game,  
the cops choose some set of edges to block and right after the robber is obliged 
to move to another vertex traversing at most $s$ unblocked edges ($s$ can be seen as the speed of the robber).
Both parts have complete knowledge of the opponent's moves
and the cops win when they occupy all edges incident to the robbers position.
We introduce the capture cost on $G$ against a robber of speed $s$. This defines 
a hierarchy of invariants, namely $\delta^{1}_{\rm e},\delta^{2}_{\rm e},\ldots,\delta^{\infty}_{\rm e}$, where  $\delta^{\infty}_{\rm e}$  is an edge-analogue of the admissibility 
graph invariant, namely the {\em edge-admissibility} of a graph.
We prove that the problem asking wether $\delta^{s}_{\rm e}(G)\leq k$, is polynomially 
solvable when $s\in \{1,2,\infty\}$ while, otherwise, it is {\sf  NP}-complete.
Our main result is a structural theorem for graphs of  bounded edge-admissibility.
 We prove that every graph of edge-admissibility at most $k$ can be constructed using $(\leq k)$-edge-sums,  starting from graphs  
 whose all  vertices, except possibly from one, have degree at most $k$.
 Our structural result is approximately tight in the sense that graphs generated by this construction always have  edge-admissibility at most $2k-1$. Our proofs 
 are based on a precise structural characterization of the graphs that do not contain $\theta_{r}$ as an immersion, where $\theta_{r}$ is the graph on two vertices and $r$ parallel edges.
  \end{abstract}


\section{Introduction}

All graphs in this paper are undirected, finite, loopless, and may have parallel edges. We denote by $V(G)$ the set of vertices of a graph $G$, while we denote by  $E(G)$  the multi-set of its edges.  We also use the term {\em $s$-path of $G$} for a path of $G$ that has length at most $s$.

A {\em $(k,s)$-hide out} in a graph $G$ is a subset $S$ of its vertices such that, for each vertex $v\in S$, it is not possible to block all $s$-paths from $v$ to the rest of $S$  by less than $k$  vertices, different than $v$. The {\em $s$-degeneracy} of a graph $G$, has been introduced in~\cite{RicherbyT11sear} as 
the minimum $k$ for which $G$ contains a $(k,s)$-hide out. $s$-degeneracy  defines a hierarchy of graph invariants that, when $s=1$, gives the classic invariant of graph degeneracy~\cite{Matula63,BodlaenderWK06con,KirousisT96}  and,  when $s=\infty$, gives the parameter of $\infty$-admissibility that was introduced by Dvořák in~\cite{Dvorak13cons} and studied in~\cite{Dvorak12astr,KiersteadT91plan,NesetrilM09frat,ChenS93grap,NesetrilM12spar,Weissauer19onth,GroheKRSS15colo}. \medskip

In this paper we introduce and study the edge analogue of the above hierarchy of graph invariants, namely the  {\em $s$-edge-degeneracy hierarchy}.  The new parameter results from the one of $s$-degeneracy if we replace   $(k,s)$-hide outs by {\em $(k,s)$-edge hide outs} where we ask that, for each vertex $v$ of $S$, it is not possible to block all $s$-paths from $v$ to the rest of $S$  by less than $k$  edges.
It follows that  the value of  $s$-edge-degeneracy
may vary considerably  than the one of $s$-degeneracy.  
For instance, consider the graph $\theta_{k}$ consisting 
of two vertices and $k$ parallel edges between them.
It is easy to see that, for every positive integer $s$,  the 
$s$-degeneracy of $\theta_k$ is 2, while it $s$-edge-degeneracy is $k$
(the two vertices form a $(k,s)$-edge hideout). In other words,  $s$-edge-degeneracy can be seen 
as an alternative way to extent the notion of degeneracy using edge separators instead of vertex separators.
\medskip

In  \autoref{magnetic}
we introduce two alternative  definitions for  $s$-edge-degeneracy, apart from the one using $(k,s)$-edge hide outs. The first is in terms of a graph searching game 
and the second is in terms of graph layouts. Next, we prove a min-max theorem 
supporting the equivalence of the three definitions. As a consequence of this theorem,
we can  identify the computational complexity of $s$-edge-degeneracy: it can be computed in polynomial time when $s\in\{1,2,\infty\}$, while for all other values of $s$, 
desiding whether its value is at most $k$
is an {\sf NP}-complete problem.\medskip

Our next step is to provide a structural theorem for  the $\infty$-edge-degeneracy that, from now on,  we call $\infty$-edge-admissibility. For $\infty$-degeneracy (also known as $\infty$-admissibility), Dvořák proved  the following structural characterization~\cite[Theorem~6]{Dvorak12astr}.

\begin{proposition}
\label{asdfasdfsdfsdfa}
For every $k$, there exist constants $d_k$, $c_{k}$ and $a_k$ such that every graph $G$ with $\infty$-admissibility at most $k$ can be constructed by  applying $(\leq c_k)$-clique sums starting from 
graphs where at most $d_{k}$ vertices have degree at least $a_{k}$.
%
\end{proposition}
In the above proposition the {\em $(\leq k)$-clique sum} operation receives as input  two graphs $G_{1}$ and $G_{2}$
such that each $G_{i}$ contains a clique $K_{i}$ with vertex set $\{v_{1}^{i},\ldots, v_{\rho}^{i}\}, \rho\leq k$. The outcome of the operation is the graph occurring if we identify 
$v_{j}^{1}$ and $v_{j}^{2}$ for $j\in\{1,\ldots,\rho\}$ and then remove some of the edges between the identified vertices. 
While the constants of \autoref{asdfasdfsdfsdfa} where not specified 
in~\cite{Dvorak12astr}, an alternative proof was recently given by Weißauer in~\cite{Weissauer19onth} where
$d_{k}=k$, $c_k=k$,  and $a_{k}=2k(k-1)$.\medskip

In \autoref{zfsdfsdfasdfsdafasd}  we provide a counterpart of \autoref{asdfasdfsdfsdfa} for  the $\infty$-{\em edge}-{\em admissibility} that is the following

\begin{theorem}
\label{presednsts}
For every $k$, every graph $G$ with $\infty$-{\em edge}-{\em admissibility} at most $k$ can be constructed by  applying $(\leq k)$-edge sums starting from 
graphs where at most one vertex has  degree at least $k+1$. 
\end{theorem}

Observe that  \autoref{presednsts} occurs from \autoref{asdfasdfsdfsdfa} if  we replace $\infty$-{\em admissibility}    by  $\infty$-{\em edge}-{\em admissibility},
if, instead of clique sums, we consider edge sums, and if we set $d_{k}=1$, $c_k=k$, and $a_{k}=k+1$. The $(\leq k)$-edge sum operation (the  definition is postponed to \autoref{asdfsdasfdasdfsdgfdgdfshdghf}) was defined  in~\cite{Wollan15thes} (see also~\cite{GiannopoulouKT15forb})
and can be seen as the  edge-counterpart of clique sums.

The proof of our structural theorem is derived by a {\sl precise} structural 
characterization of the graphs where each pair of vertices is separated by a cut of size at most $k$.
We prove that these graphs are {\sl exactly} those that can be constructed using $(\leq k)$-edge sums from 
graphs where all but one of their vertices have degree at most $k$. This directly implies 
our structural theorem for $\infty$-{\em edge}-{\em admissibility}, as every pair of two vertices linked by $k+1$ pairwise edge-disjoint paths is a $(k+1,\infty)$-edge hide out. 

Our last result is that the converse of the structural characterization in \autoref{presednsts} holds in an approximate way: if $G$ can be constructed using $(\leq k)$-edge sums from 
graphs where all but one of their vertices have degree at most $k$, then 
the  $\infty$-{edge}-{admissibility} of $G$ is at most $2k-1$. This suggests that our decomposition theorem is indeed the correct choice for the parameter of  $\infty$-{edge admissibility}.

\section{Basic definitions}

\textbf{Sets and integers.}
Given a non-negative integer $s$, we denote by $\Bbb{N}_{\geq s}$ the set of all non-negative integers that are not smaller than $s$.  We also denote $\Bbb{N}^{+}_{\geq s}=\Bbb{N}_{\geq s}\cup\{\infty\}$.
Given two integers $p\leq q,$  we set $\intv{p,q}=\{p,p+1,\ldots,q\}$
and given a $k\in\Nbb_{\geq 0}$ we define $[k]=[1,k].$ Given a set $A,$ we use $2^{A}$ for the set of all its subsets, we define
$\binom{A}{2}:=\{S\mid S\in 2^{A} \mbox{~and~} |S|=2\},$ and, given a $k\in\Bbb{N}_{\geq 0}$
we  denote by $A^{(\leq k)}$ the set of 
all subsets of $A$ that have size at most $k$. A {\em near-partition} of a set $A$ is a collection of pairwise disjoint sets whose union is $A$. A {\em bipartition} of $A, |A|\geq 2$ is a near-partition of $A$ into two 
non-empty sets.

\paragraph{Graphs.} All graphs in this paper are undirected, finite, loopless, and may have parallel edges. We denote by $V(G)$ the set of vertices of a graph $G$ while we use $E(G)$ for the multi-set of its edges.
Given a graph $G$ and a vertex $v$, we define $E_{G}(v)$ as the multi-set of all edges
of $G$ that are  incident to $v$. We  define the {\em neighborhood} of  $v$ as $N_{G}(v)=(\bigcup_{e\in  E_{G}(v)}e)\setminus\{v\}$,  the {\em edge-degree} of $v$ as $\deg_{G}(v)=|E_{G}(v)|$. We also define $\Delta(G)=\max\{\deg_{G}(v)\mid v\in V(G)\}$. 
Given a $F\subseteq E(G)$, we define $G\setminus F=(V(G),E(G)\setminus F)$.

Given a tree $T$ and two vertices $a,b\in V(T)$ we define $aTb$ as the path
of $T$ connecting $a$ and $b$.
Let $G$ be a graph and let $S_{1},S_{2}\subseteq V(G)$ where $S_{1}\cap S_{2}=\emptyset$. We define $$E_{G}(S_{1},S_{2})=\{e\in E(G)\mid e\cap V_{1}\neq\emptyset \mbox{~and~}e\cap V_{2}\neq\emptyset\}.$$
A {\em cut} of a graph $G$ is any  bipartition $(X,\overline{X})$ of its vertices. 
The {\em edges} of a cut $(X,\overline{X})$
is the set $E(X,\overline{X})$ while the {\em size}
of $(X,\overline{X})$ is equal to $|E(X,\overline{X})|$.
Given two distinct vertices $x$ and $y$ of $G$, an 
{\em $(x,y)$-cut} of $G$ is a 
cut $(X,\overline{X})$ of $G$ such that $x\in X$ and $y\in \overline{X}$.

We define the function $\rho: 2^{V(G)}\to\Bbb{N}$ such that 
$\rho(X)=|E_{G}(X,V(G)\setminus X)|$. It is easy to see that $\rho$ is a submodular function, ie., 
\begin{eqnarray}
\forall X,Y\in 2^{V(G)} & \rho(X\cap Y)+\rho(X\cup Y)\leq \rho(X)+\rho(Y).
\end{eqnarray}

Given a graph $G$ and two distinct $x,y\in V(G)$, we call an {\em $(x,y)$-$s$-path}
every $s$-path in $G$ starting from $x$ and finishing on $y$. We also use the term {\em $(x,y)$-path} as a 
shortcut for $(x,y)$-$\infty$-path.
We define the function ${\bf  cut}_{G,s}: {V(G)\choose 2}\to \Bbb{N}_{\geq 0}$ so that 
${\bf  cut}_{G,s}(x,y)$ is equal to the minimum size of a $F\subseteq E(G)$ such that $G\setminus F$ 
does not contain any $(x,y)$-$s$-path.
The complexity of   computing ${\bf  cut}_{G,s}(x,y)$ 
is provided   by the next proposition (see~\cite{BaierEHKSS06len,MahjoubM10maxf,ItaiPS82thec}).

\begin{proposition}
\label{writings}
If $s\in\{1,2,\infty\}$, then the problem that, given a graph $G$, a $k\in \Bbb{N}$, and  two distinct vertices $a$ and  $b$ of $G$, asks whether ${\bf cut}_{G,s}(a,b)\leq k$  
is polynomially solvable, while it is  {\sf  NP}-complete if $s\in\Bbb{N}_{\geq 3}$.
\end{proposition}

%

\section{Graph searching and $s$-edge-degeneracy}
\subsection{A search game.}
\label{magnetic}

The study of graph searching parameters is an active field of graph theory. Several 
important graph parameters have their search-game analogues that provide useful
insights on their combinatorial and algorithmic properties. 
(For related surveys, see~\cite{FominT08anan,Bienstock89-Gra,FominP96-Pu,Alspach04-Se,AlpernG03-Th}.) 

We introduce a graph searching game, where the opponents 
are a group of cops and a robber. In this game,  
the cops are blocking {\sl edges} of the graph, while the robber resides on the {\sl vertices}.
The first move of the game is done by the robber, who chooses 
a vertex to occupy. Then, the game is played in rounds. In each round, 
 first the cops   block a set of edges and next  the robber moves 
to another vertex via a path consisting of at most $s$ unblocked edges. The robber
 {\sl cannot 
stay put} and he/she is captured if, after the move of the cops, all the edges incident to his/her current location are blocked. Both cops and robbers have full 
knowledge of their opponent's current position and they take it into consideration before they make their next move.  We next give the formal definition of the game.

The game is parameterized by the speed $s\in\Bbb{N}^{+}_{\geq 1}$ of the robber. A {\em search strategy on $G$} for the cops is a function $f: V(G)\rightarrow 2^{E(G)}$
that, given the current position $x\in V(G)$ of the robber in the end of a round, outputs the set $f(v)$ of the edges that should be blocked in the beginning of the next round.
The {\em cost} of a cop strategy $f$ is defined as ${\bf cost}(f)=\max\{|f(v)|\mid v\in V(G)\}$, i.e., the maximum number of edges that may be blocked by the robbers according to $f$.

An {\em escape strategy on} $G$ 
for the robber is a pair $R=(v_{\rm start},g)$ where $v_{\rm start}$ is the vertex 
of robber's first move and $g: 2^{E(G)}\times V(G)  \rightarrow V(G)$ is a function that, given the set $F$ of blocked edges in the beginning of a round
and the current position $x$ of the robber, outputs the vertex $u=g(F,v)$ where the robber should move. 
Here the natural restriction for $g$ is that there is an $s$-path from 
$v$ to $u$ in $G\setminus F$. 
Clearly, if $F$ is the set of edges that are incident to $v$, then $g(F,v)$ should 
be equal to $v$ and this expresses the situation where the robber is captured.

Let $f$ and $R=(v_{\rm start},g)$ be strategies for the cop and the robber respectively.
The  {\em game scenario}  generated by the pair $(f,R)$ is the infinite sequence $v_0,F_{1},v_{1},F_{2},v_{2},\ldots,$ where $v_{0}=v_{\rm start}$ and for every $i\in\Bbb{N}_{\geq 1}$, $F_{i}=f(v_{i-1})$ and $v_{i}=g(F_{i},v_{i-1})$. If $v_{i}=v_{i-1}$ for some $i\in \Bbb{N}_{\geq 1}$,
then $(f,R)$ is a {\em cop-winning} pair, otherwise it is a {\em robber-winning} pair.

The {\em capture cost against a robber of speed $s$} in a graph $G$, denoted by ${\bf cc}_{s}(G)$  is the minimum $k$ for which there is a 
cop strategy $f$, of cost at most $k$, such that for every robber strategy $R$, $(f,R)$
is a cop-winning pair.

\subsection{A min-max theorem for $s$-edge-degeneracy}

\paragraph{$s$-edge-degeneracy.}

Let $G$ be a graph, $x\in V(G)$,  $S\subseteq V(G)\setminus \{x\}$, and $s\in\Bbb{N}^{+}_{\geq 1}$. 
We say that a set $A\subseteq E(G)$ is an
\emph{$(s, x, S)$-edge-separator} if every $s$-path of $G$ from $x$ to some vertex in $S$, contains some edge from $A$.
 We  define ${\bf supp}_{G,s}(x,S)$ to be the minimum size
of  an $(s, x, S)$-edge-separator in $G$. 
\medskip

Let $G$ be a graph and let $L=\langle v_{1},\ldots,v_{r} \rangle$ be a layout (i.e. linear ordering) of its vertices. Given an $i\in[r],$ we denote $L_{\leq i}=\langle v_{1},\ldots,v_{i} \rangle$. Given an $s\in\Bbb{N}^{+}_{\geq 1}$, we define the {\em $s$-edge-support} of a vertex $v_{i}$ in $L$ as ${\bf supp}_{G,s}(v_i,L_{\leq i-1})$.
The {\em $s$-edge-degeneracy} of $L$, is the maximum $s$-edge-support of a vertex in $L$. The {\em $s$-edge-degeneracy} of $G$, denoted by $\delta_{\rm e}^{s}(G)$ is the minimum $s$-edge-degeneracy over all layouts of $G$.

\paragraph{$(k,s)$-edge-hide-outs.}
 Let $s\in\Bbb{N}^{+}_{\geq 1}$ and $k\in\Bbb{N}$. A \emph{$(k,s)$-edge-hide-out} in a graph $G$ is any set
$R\subseteq V(G)$ such that, for every $x\in R$, ${\bf supp}_{G,s}(x,R\setminus\{
x\})\geq k$.  A \emph{$(k,s)$-edge-hide-out} $S$ is {\em maximal} there is no other \emph{$(k,s)$-edge-hide-out} $S'$ with $S\subsetneq S'$. It is easy to verify that every graph contains a unique 
maximal $(k,s)$-edge-hide-out.

$(k,s)$-edge-hide-outs can be seen as obstructions to small $s$-edge-degeneracy. In particular we prove the following min-max theorem, characterizing the search game that we defined in \autoref{magnetic}.
\begin{theorem}
\label{consiste}
Let $G$ be a graph and let $s\in\Bbb{N}^{+}_{\geq 1}$ and $k\in\Bbb{N}$.
The following three statements are equivalent.
\begin{enumerate}

\item[(1)] ${\bf cc}_{s}(G)\leq k$, i.e., there is a cop strategy $f$ on $G$ of cost less than $k$,  such
that for every robber strategy $R$ on $G$, $(f,R)$ is cop-winning.

\item[(2)] $G$ has no $(k+1,s)$-edge-hide-out.

\item[(3)] $\delta_{\rm e}^{s}(G)\leq  k$.
\end{enumerate}
\end{theorem}

\begin{proof}
{(\em 1\,)} $\Rightarrow$ {\em (2)}. We prove that the negation of {\em (2)} implies the 
negation of {\em (1)}.
Suppose that $S$ is a  $(k+1,s)$-edge-hide-out of $G$.
We use $S$ in order to build an escape strategy $R=(v_{\rm start},g)$ on $G$ as follows: Let $v_{\rm start}$ be any vertex in $S$. 
Let now $v\in S$ and $F\in 2^{E(G)}$. If $|F|>k$, then $g(v,F)=v$.
We next define $g(v,F)$ for every $F\in E(G)^{\leq k}$.
As $S$ is a  $(k+1,s)$-edge-hide-out of $G$,
we know that ${\bf supp}_{G,s}(v,S\setminus\{
v\})\geq k+1$, therefore there is an $s$-path from $v$ to some vertex $u\in S\setminus \{v\}$  that avoids all edges in $F$.  We define $g(v,F)=u$.
Notice now that if $f$ is a cop strategy on $G$ of cost at most $k$,
and $v_0,F_{1},v_{1},F_{2},v_{2},\ldots,$  is the game scenario generated by the pair $(f,R)$, then $v_{i-1}\neq v_{i}$ for every $i\in\Bbb{N}_{\geq 1}$. This means 
that $R$ is a robber-winning strategy against any cop strategy of cost at most $k$, therefore ${\bf cc}_{s}(G)\geq  k+1$.\medskip

\noindent {(\em 2\,)} $\Rightarrow$ {\em (3)}. 
Let $n=|V(G)|$.
As $G$ has no $(k+1,s)$-edge-hide-out, it follows 
that for every $R\subseteq V(G)$ there is a vertex $v\in R$,
such that ${\bf supp}_{G,s}(v,R\setminus\{
v\})\leq  k$. We pick such a vertex for every $R\subseteq V(G)$
and we denote it by $v(R)$. We now set $V_{n}=V(G)$, $v_{n}=v(V_{n})$, and
for $i\in \langle n-1,\ldots,1\rangle$ we set 
 $V_{i}=V_{i+1}\setminus \{v_{i+1}\}$, $v_{i}=v(V_{i})$.
 We now set $L=\langle v_{1},\ldots,v_{n}\rangle$ and 
observe that for every $i\in[n]$, 
${\bf supp}_{G,s}(v_i,L_{\leq i-1})={\bf supp}_{G,s}(v_i,V_{i-1})\leq k$.
Therefore, the $s$-edge-degeneracy of $L$ is at most $k$, hence $\delta_{\rm e}^{s}(G)\leq  k$.\medskip

\noindent {(\em 3\,)} $\Rightarrow$ {\em (1)}. 
Suppose now that $L=\langle v_{1},\ldots,v_{n}\rangle$ is a layout of $V(G)$
such that, for every $i\in[n]$, ${\bf supp}_{G,s}(v_i,L_{\leq i-1})\leq k$.
We use $L$ to build a cop strategy $f: V(G)\rightarrow 2^{E(G)}$ as follows. Let $i\in [n]$
and let $F_{i}$ be an $(s, v_{i}, L_{\leq i-1})$-edge-separator of $G$. We  define $f$  by setting 
$f(v_{i})=F_{i}$.
This means that if at some point the robber occupies vertex $v_{i}$,
then there is no $s$-path in $G\setminus F_{i}$ from $v_{i}$ to $L_{\leq i-1}$.
As a consequence of this, no matter what the robber strategy $R=(v_{\rm start},g)$ is, it should hold that
$g(v_{i},F_{i})\in L_{\geq i}$.
Therefore if $x_0,F_{1},x_{1},F_{2},x_{2},\ldots,$ is the {game scenario}  generated by the pair $(f,R)$, then $x_i=x_{i-1}$ for some $i<n$.
\end{proof}

\subsection{The complexity of $s$-edge-degeneracy, for distinct values of $s$}
\label{dsfsadfdsfds}

We now combine \autoref{writings} with the min-max theorem of the previous subsection in order to identify the computational 
complexity of $\delta_{\rm e}^{s}$ for different values of $s$.  Our main result is the following.

\begin{theorem}
\label{unguided}
If $s\in\{1,2,\infty\}$, then the problem that, given a graph $G$ and a $k\in \Bbb{N}$, asks whether $\delta_{\rm e}^{s}(G)\leq  k$,  
is polynomially solvable, while it is  {\sf  NP}-complete if $s\in\Bbb{N}_{\geq 3}$.
\end{theorem}

\begin{proof}
Notice first that checking whether $\delta_{\rm e}^{s}(G)\leq k$ can be done by 
the algorithm \textbf{check $s$-edge degeneracy} in \autoref{narcotic}. Indeed, if the  
maximal  $(k+1,s)$-edge-hideout $S$ is non-empty then the above algorithm will report that  $\delta_{\rm e}^{s}(G)> k$ after visiting, in line 3, every vertex not in $S$, as, by  the maximality of $S$,  for every  $S'\supsetneq S$ there is a vertex $x\in S'\setminus S$ where  ${\bf supp}_{G,s}(x,S'\setminus\{x\})\leq k$. On the other hand, if $S$ is empty, then the procedure will produce 
a layout $L=\langle v_1,\dots,v_n\rangle$ with  $s$-edge-degeneracy at most $k$.

\begin{figure}[h]
{\small
\begin{tabbing}
{\bf Algorithm} \textbf{check $s$-edge degeneracy}\\[-3mm]
    {\sl Output:} \=\kill \\
    {\sl Input:}  \>a graph $G$ and an integer $k\in\Bbb{N}_{\geq 0}$. \\
    {\sl Output:} \>a report on whether $\delta^{s}_{\rm e}(G)\leq k$.\\[2mm]
    1. \=  $n\leftarrow |V(G)|$, $S\leftarrow V(G)$.                                               \\
    2. \> for \= $ i=n,\dots,1$,                                                \\
    3. \> \> if \=  there is an $x\in S$ with ${\bf supp}_{G,s}(x,S\setminus\{x\})\leq k$ then
        $v_{i}\leftarrow x$,                                          \\
     \> \> else \= report that ``$\delta^{s}_{\rm e}(G) > k$'' and {\bf stop}\\
    \> \> \>  /\!\!/ \= $S$ is the maximal $(k+1,s)$-edge-hideout of $G$,\\
    \>\>  \>  \> witnessing that $\delta^{s}_{\rm e}(G) > k$, because of \autoref{consiste}./\!\!/ \\
    4.  \> \>  $S\leftarrow S-v_{i}$.                                        \\
    5. \> Output ``$\delta^{s}_{\rm e}(G)\leq k$, witnessed by layout $L=\langle v_1,\dots,v_n\rangle.$''
\end{tabbing}
}\vspace{-4mm}
\caption{An algorithm checking whether $\delta^{s}_{\rm e}(G)\leq k$.}
\label{narcotic}
\end{figure}

Clearly, \textbf{check $s$-edge degeneracy} runs in polynomial time if checking whether ${\bf supp}_{G,s}(x,S\setminus\{x\})\leq k$ can be done in polynomial time, which is 
equivalent to  checking whether ${\bf cut}_{G’,s}(x,x’)\leq k$ where $G'$ is the graph obtained by $G$ after we identify all vertices of $S\setminus\{x\}$ to a single vertex $x'$.
As this is possible for $s\in\{1,2,\infty\}$, due to \autoref{writings}, the polynomial part of the theorem follow.
\medskip

It now remains to prove that checking whether $\delta_{\rm e}^{s}(G)\leq k$ is an {\sf NP}-hard problem when $s\in\Bbb{N}_{\geq 3}$.
For this we will reduce the problem of checking whether ${\bf cut}_{G,s}(a,b)\leq k$
to the problem of checking whether $\delta_{\rm e}^{s}(G)\leq  k$ and the result will follow from the hardness part of  \autoref{writings}.

Let ${\bf T}_{s}=(G,a,b,k)$ be a quadruple where $G$ is a graph on $n$ vertices,  $k\in \Bbb{N}_{\geq 0}$, and $a$, $b$ two distinct vertices of $G$. 
We construct the 
graph $G_{{\bf T}_{s}}$ as follows:  Take $k+n+1$ copies $G_{1},\ldots,G_{k+n+1}$
of $G$ and  identify all $a$'s of these copies to a single vertex that we call again $a$, while we set   $B:=\{b_{1},\ldots,b_{k+n+1}\}$ where $b_i$ is the copy of $b$ in $G_{i}$. Next, we add $n$ new vertices $C=\{c_{1},\ldots,c_{n}\}$
and, for every $(i,j)\in [n]\times [k+n+1]$,  we add the edge $e_{i,j}=c_ib_j$.
The construction of $G_{{\bf T}_{s}}$ is completed by subdividing each edge $e_{i,j}$ $s-1$ times. 

For every $(i,j)\in [n]\times [k+n+1]$, we denote by $P_{i,j}$ the $(c_i,b_j)$-$s$-path that replaces $e_{i,j}$ after this subdivision.
Also we set 
$${\cal P}_{j}=\{P_{i,j}\mid i\in [n]\}, \mbox{ for } j\in[k+n+1],$$ 
$${\cal Q}_{i}=\{P_{i,j}\mid j\in[k+n+1]\}, \mbox{ for } i\in[n],$$ and 
$P=\bigcup_{j\in[k+n+1]}{\cal P}_{j}$.\smallskip

For the correctness of the reduction, it remains to prove the following.
\begin{eqnarray}
\delta_{\rm e}^{s}(G_{{\bf T}_{s}})\leq k+n & \iff &  {\bf   cut}_{G,s}(a,b)\leq k\label{consists}
\end{eqnarray}
We first claim that, for every $j\in[k+n+1]$, 
\begin{eqnarray}
{\bf   cut}_{G,s}(a,b)=  {\bf  cut}_{G_{{\bf T}_{s}},s}(a,b_j).\label{appetite}
\end{eqnarray}
To see~\eqref{appetite} observe that none of the $(b_j,a)$-$s$-paths of $G_{{\bf T}_{s}}$
contains any vertex outside $G_{j}$, therefore ${\bf  cut}_{G_{{\bf T}_{s}},s}(a,b_j)={\bf  cut}_{G_{j},s}(a,b_j)={\bf   cut}_{G,s}(a,b)$.\medskip

We first prove the ($\Rightarrow$) direction of the \eqref{consists}.
For this we assume that  ${\bf   cut}_{G,s}(a,b)\geq  k+1$ and we show 
that $G_{{\bf T}_{s}}$ contains a $(k+n+1,s)$-edge-hide-out, which, by \autoref{consiste}, yields 
$\delta_{\rm e}^{s}(G_{{\bf T}_{s}})\geq  k+n+1$.
We claim that $S:=C\cup B\cup \{a\}$ is a   $(k+n+1,s)$-edge-hide-out of $G_{{\bf T}_{s}}$.  
As ${\bf   cut}_{G,s}(a,b)\geq  k+1\geq 1$, we know that for each $j\in[k+n+1]$ there
is a $(b_j,a)$-$s$-path, say  $R_{j}$, in $G_{{\bf T}_{s}}$ whose internal vertices 
are not vertices of any path in $P$. Moreover, every two paths in ${\cal R}:=\{R_{j}\mid j\in[k+n+1]\}$ have only one vertex, that is $a$ in common. The fact that $|{\cal R}|=k+n+1$ implies that  ${\bf cut}_{G_{{\bf T}_{s}},s}(a,B)\geq k+n+1$. Therefore, as  ${\bf cut}_{G_{{\bf T}_{s}},s}(a,S\setminus \{a\})\geq {\bf cut}_{G_{{\bf T}_{s}},s}(a,B),$ we have that 
\begin{eqnarray}
{\bf cut}_{G_{{\bf T}_{s}},s}(a,S\setminus \{a\})\geq k+n+1.\label{adorning}
\end{eqnarray}
Consider now the vertex $b_{j}$, for some $j\in[k+n+1]$, and notice that  that ${\bf cut}_{G_{{\bf T}_{s}},s}(b_j,W\cup \{a\})\geq {\bf  cut}_{G_{{\bf T}_{s}},s}(a,b_j)+|{\cal P}_{j}|$. Combining this with~\eqref{appetite} and the fact that $|{\cal P}_{j}|=n$, we obtain that 
${\bf cut}_{G_{{\bf T}_{s}},s}(b_j,W\cup \{a\})\geq
{\bf   cut}_{G,s}(a,b)+n\geq k+n+1$. As ${\bf cut}_{G_{{\bf T}_{s}},s}(b_j,S\setminus \{b_{j}\})\geq {\bf cut}_{G_{{\bf T}_{s}},s}(b_j,W\cup \{a\}),$ we have that 
\begin{eqnarray}
\forall j\in[k+n+1]\quad {\bf cut}_{G_{{\bf T}_{s}},s}(b_j,S\setminus \{b_{j}\})\geq k+n+1.\label{skeleton}
\end{eqnarray}

Consider now the vertex $c_{i}$, for some $i\in[n]$. Notice that ${\bf cut}_{G_{{\bf T}_{s}},s}(c_{i},B)\geq |{\cal Q}_{i}|=k+n+1$. As ${\bf cut}_{G_{{\bf T}_{s}},s}(c_{i},S\setminus \{c_{i}\})\geq {\bf cut}_{G_{{\bf T}_{s}},s}(c_{i},B)$ we obtain that 
\begin{eqnarray}
\forall i\in[n]\quad {\bf cut}_{G_{{\bf T}_{s}},s}(c_i,S\setminus \{c_{i}\})\geq k+n+1.\label{friesian}
\end{eqnarray}
It now follows from~\eqref{adorning},~\eqref{skeleton}, and~\eqref{friesian}, that $S$ is an $(k+n+1,s)$-edge-hide-out of $G_{{\bf T}_{s}}$, as required.
\medskip

We now prove the  ($\Leftarrow$) direction of ~\eqref{consists}. The assumption that ${\bf   cut}_{G,s}(a,b)\leq k$ implies that ${\bf  cut}_{G_{{\bf T}_{s}},s}(a,b_j)\leq k$, because of~\eqref{appetite}. Therefore there is a set $F_{j}$ of edges in $G_{i}$
that blocks every $(b_{j},a)$-$s$-path of $G_{{\bf T}_{s}}$.

Let $L=\langle v_{1},\ldots,v_{\ell}\rangle$ be any layout of the vertices of $G_{{\bf T}_{s}}$ where 
\begin{eqnarray}
L_{\leq k+2n+2}=\langle a,c_{1},\ldots,c_{n},b_{1},\ldots,b_{k+n+1}\rangle \label{granting}
\end{eqnarray}
In order to prove that  $\delta_{\rm e}^{s}(G_{{\bf T}_{s}})\leq k+n$ it suffices to show that, for each $h\in[\ell]$, 
\begin{eqnarray}
 {\bf supp}_{G,s}(v_h,L_{\leq h-1})  & \leq & k+n.\label{segments}
\end{eqnarray}

Notice that the  vertices of $L_{\leq k+2n+2}$ are the vertices of $S=W\cup B\cup\{a\}$.
As each other vertex $v\in V(G)\setminus S$, has degree at most $n-1$ in $G_{{\bf T}_{s}}$,
we directly have that \eqref{segments} holds when  $h\in [k+2n+3,\ell]$.
Let now $v_{h}=b_{j}$ for some $j\in[k+n+1]$. Let $F_{j}^*$ be the edges incident to $b_{j}$
that are edges of the paths in ${\cal P}_{j}$. Observe that $F_{j}\cup F_{J}^{*}$ blocks 
in $G_{{\bf T}_{s}}$ all the $s$-paths from $L_{\leq h-1}$ to $b_{j}$.
As all the edges in  $F_{j}\cup F_{J}^{*}$ have some endpoint in $L_{\geq h}$
and $|F_{j}|+|F_{j}^{*}|\leq k+n$, we conclude that  \eqref{segments} holds when  $h\in[n+2,k+2n+2]$.
Let now $v_{h}=c_{i}, i\in [n]$. Notice that the distance in $G_{{\bf T}_{s}}$ 
between $c_{i}$ and any vertex in $\{a\}\cup (W\setminus \{c_{i}\})$ is bigger than $s$, therefore 
${\bf supp}_{G,s}(v_h,L_{\leq h-1})\leq |F_{j}|+|F_{j}^{*}|\leq k+n$ and  \eqref{segments} holds when  $h\in[2,n+1]$. Finally \eqref{segments} holds trivially when $h=1$. This completes the proof of  ~\eqref{consists}, and the theorem follows.
\end{proof}

%
%
%
%

\section{A structural theorem for edge-admisibility}
\label{zfsdfsdfasdfsdafasd}

This section is dedicated to the statement and proof of our structural characterization for $\delta_{\rm e}^{\infty}$.
%
%
%
%

%
%
%


\subsection{Basic definitions}
\label{asdfsdasfdasdfsdgfdgdfshdghf}

\paragraph{Edge-admissibility}
The {\em $\infty$-admissibility} of a graph $G$  is the minimum $k$
for which there exists a layout $L=\langle v_{1},\ldots,v_{n}\rangle$ of $V(G)$ such that for every $i\in[n]$ there are at most  
$k$ vertex-disjoint, except for $v_i$, paths from $v_{i}$ to $L_{\leq i-1}$ in $G$. 
If in this definition we replace  ``vertex-disjoint'' by ``edge-disjoint'' (and we obviously drop the exception of $v_i$)
we have an edge analogue of the admissibility invariant that, because of Menger's theorem is the 
same invariant as $\delta_{\rm e}^{\infty}$. This encourages us to alternatively refer to $\delta_{\rm e}^{
\infty}(G)$ as the {\em $\infty$-edge-admissibility} of the graph $G$.

The purpose of this section is to give a structural characterization for graphs of bounded edge-admissibility. For this we need first a series of definitions. 
\paragraph{Immersions.}
Given a graph $G$ and two  incident edges $e$ and $f$  of $G$ (i.e., edges with a common endpoint) 
the result of {\em lifting $e$ and $f$ in $G$} is the graph obtained from $G$ after removing $e$ and $f$ and then adding the 
edge formed by the symmetric difference of $e$ and $f$. We say that a graph $H$ is an {\em immersion}
of a graph $G$, denoted by $H\leq G$,  if a graph isomorphic to $H$ can be obtained 
from some subgraph of $G$ after a series of liftings of incident edges. 
Given a graph $H$, we define the class of {\em $H$-immersion free graphs}
as  the class of all graphs that do not  contain $H$ as an immersion.

\paragraph{Edge sums.} 
Let $G_{1}$ and $G_{2}$ be graphs, let $v_{1},v_{2}$ be vertices of $V(G_{1})$ and $V(G_{2})$ respectively such that $k=\deg_{G}(v_{1})=\deg_{G}(v_{2})$, 
and consider a bijection $\sigma: E_{G_{1}}(v_{1})\rightarrow E_{G_{2}}(v_{2})$, 
where $E_{G_{1}}(v_{1})=\{e_{1}^{i}\mid i\in [k]\}$.  We define the {\em $k$-edge sum of 
$G_{1}$ and $G_{2}$ on $v_{1}$ and $v_{2}$, with respect to $\sigma$}, as the graph $G$ obtained if we take 
the disjoint union of $G_{1}$ and $G_{2}$, identify $v_{1}$ with $v_{2}$, and
then,  for each $i\in\{1,\ldots,k\}$, lift $e_{1}^{i}$ and $\sigma(e_{1}^{i})$ to 
a new edge $e^{i}$ and remove the vertex $v_{1}$.
We say that $G$ is a {\em $(\leq k)$-edge sum of $G_1$ and $G_{2}$} 
if either $G$ is the disjoint union of $G_{1}$ and $G_{2}$
or there is some $k'\in [k]$, two vertices $v_{1}$ and $v_{2}$, and a bijection $\sigma$ as above such that $G$ is the $k'$-edge sum of 
$G_{1}$ and $G_{2}$ on $v_{1}$ and $v_{2}$, with respect to $\sigma$.

\begin{figure}[ht]
\begin{center}
{\includegraphics[width=120mm]{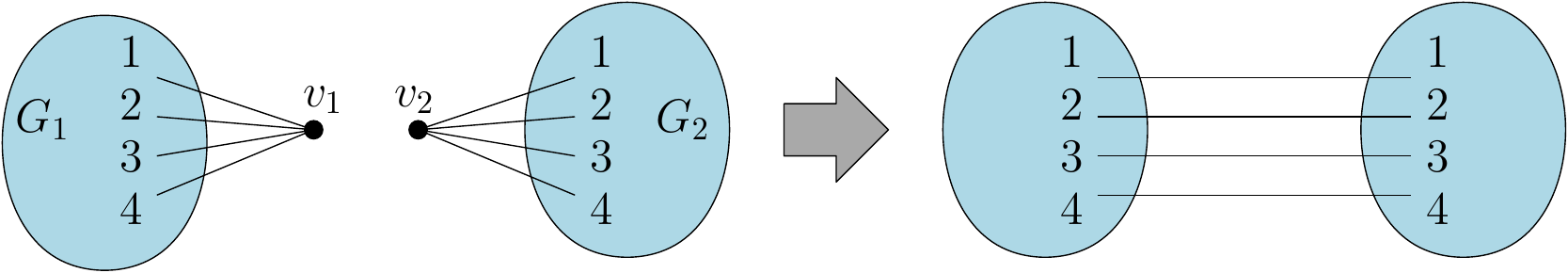}}
\end{center}
\caption{The graphs $G_{1}$ and $G_{2}$ and the graph created after the  edge-sum of $G_{1}$ and $G_{2}$.}
\label{pensions}
\end{figure}

Let ${\cal G}$ be some graph class. We recursively define the {\em $(\leq k)$-sum 
closure of} ${\cal G}$, denoted by  ${\cal G}^{(\leq k)}$,
as the set of graphs containing every graph $G\in {\cal G}$ that is the  $(\leq k)$-edge sum of two graphs $G_{1}$ and $G_{2}$ in ${\cal G}$ where $|V(G_{1})|,|V(G_{2})|< |V(G)|$.\medskip

A graph $G$ is {\em almost $k$-bounded edge-degree} if all its vertices, except possibly from one, have edge-degree at most $k$. We denote this class of graphs by ${\cal A}_{k}$.

The rest of this section is devoted to the proof of the the following result.

\begin{theorem}
\label{presents}
For every graph $G$ and $k\in \Bbb{N}_{\geq 0}$, if $G$ has edge-admissibility at most $k$, then $G$ can be constructed by  almost $k$-bounded edge-degree graphs after a series of $(\leq k)$-edge sums, i.e., $G\in {\cal A}_{k}^{(\leq k)}$.
Conversely, for every $k\in \Bbb{N}_{\geq 1}$, every graph in ${\cal A}_{k}^{(\leq k)}$ has edge-admissibility 
at most $2k-1$.
\end{theorem}

%

%
%
%
%
%

\subsection{A structural characterizations of $\theta_{k}$-immersion free graphs}

Recall that given a $k\in\Bbb{N}_{\geq 1}$, $\theta_{k}$ is the graph with two vertices and $k$ parallel edges between them. In this subsection we prove that $\theta_{k}$-immersion free graphs are exactly the graphs in ${\cal A}_{k}^{(\leq k)}$ (\autoref{furrowed}).\medskip

 We need some more definitions in order to 
translate edge-sums to their decomposition equivalent that will be more easy to handle.

\paragraph{Tree-partitions.}
A {\em tree-partition} of a graph $G$  is a pair ${\cal D}=(T,{\cal B})$ where $T$ is 
a tree and ${\cal B}=\{B_{t}\mid t\in V(T)\}$ is a near-partition of $V(G)$. We refer to the sets 
in ${\cal B}$ as the {\em bags} of ${\cal D}$.
Given a tree-partition ${\cal D}=(T,{\cal B})$ of $G$ and an edge $e\in E(T)$, we define ${\bf  cross}_{\cal D}(e)=E_{G}(V_{1},V_{2})$,  where $V_{i}=\bigcup_{t\in V(T_{i})}B_{t}$, for $i\in[2]$
and $T_{1}$ and $T_{2}$ are the two connected components of $T\setminus e$.

\begin{figure}[ht]
\begin{center}
{\includegraphics[width=14cm]{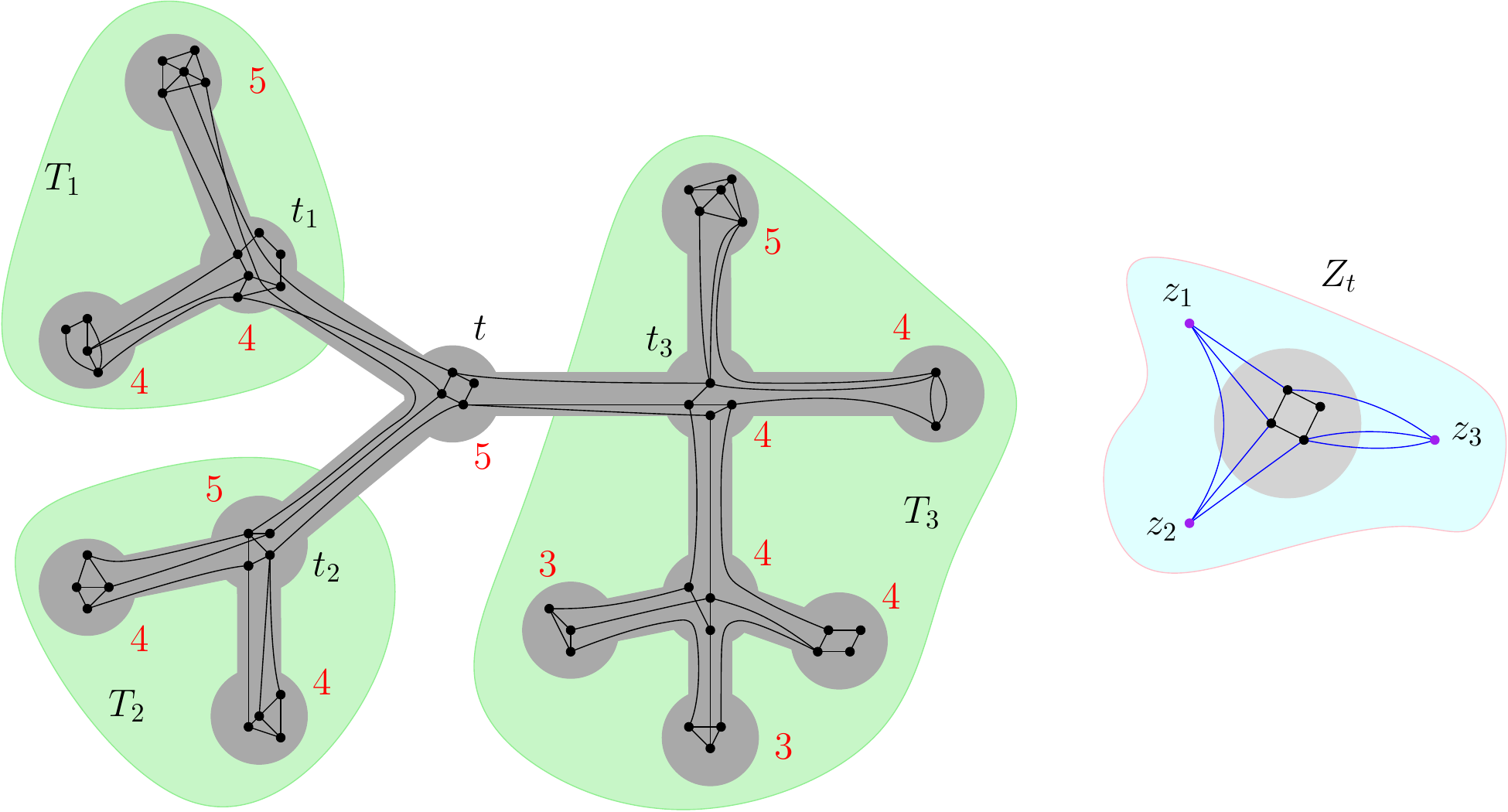}}
\end{center}
\caption{A graph $G$, a tree-partition of $G$ with  adhesion 3, and the torso $Z_{t}$ of the vertex $t$.}
\label{pensasions}
\end{figure}

For each $t\in V(T)$, we define the $t$-torso of ${\cal D}$  as  follows: Let $T_{1},\ldots,T_{q_t}$ be the connected components
of $T\setminus t$ and  let $t_{1},\ldots,t_{q_{t}}$ be the neighbors of $t$ in $T$ such that $t_{i}\in V(T_{i})$. We set  $\bar{B}_{i}=\bigcup_{h\in V(T_{i})}B_{h}$, for $i\in[q_{t}]$.
Νext, we define the graph $Z_{i}$ as the graph obtained from $G$
if, for every $i\in[q_{t}],$ we identify all the vertices of $\bar{B}_{i}$
to a single vertex $z_{i}$ (maintaining the multiple edges created after such an identification). We call  $Z_{t}$  {\em the $t$-torso of ${\cal D}$} or, simply {\em a torso} of ${\cal D}$. We call the new vertices $z_{1},\ldots,z_{q_{t}}$ {\em satellites} of the torso $Z_{t}$. 
For each $i\in[q_{t}]$,  we say that $z_{i}$ {\em represents} the vertex $t_{i}$ in $T$
and {\em subsumes} the connected component $T_{i}$ of $T\setminus t$.
For an example of a tree-partition, see~\autoref{pensasions}.\smallskip

Let ${\cal D}=({\cal B},T)$ be a tree-partition of a graph $G$.
The {\em adhesion} of ${\cal D}=(T,{\cal B})$ is  $ \max\{|{\bf  cross}_{\cal D}(e)|\mid e\in E(T)\}$ (the adhesion of the tree-partititon of \autoref{pensasions} is 3).
The {\em strength} of ${\cal D}=(T,{\cal B})$ is  $\min\{\Delta(Z_{t})\mid t\in V(T)\}$
(in the tree-partititon of \autoref{pensasions} the red numbers are the values of $\Delta(Z_{t})$ for each node of the tree $T$).

Observe that if ${\cal D}$ has strength at least $k+1$, then 
every torso of ${\cal D}$ contains a vertex of degree at least $k+1$.

Notice that each graph $G$, where $\Delta(G)\leq k$, has a tree-partition $(T,{\cal B})$ where both adhesion and strength are at  most $k$: let $T$ be a star with center $r$ and $|V(G)|$ leaves $\ell_{1},\ldots,\ell_{|V(G)|}$, consider a numbering $v_{1},\ldots,v_{|V(G)|}$ of $V(G)$, 
and then set $B_{r}=\emptyset$, while $B_{\ell_{i}}=\{v_{i}\}, i\in[|V(G)|]$.
\medskip

The next observation follows directly from the definitions and provides a ``translation'' of edge-sums 
in terms of tree-partitions. 

\begin{observation}
\label{teaching}
Let ${\cal G}$ be a graph class and let $k\in \Bbb{N}$. The
class ${\cal G}^{(\leq k)}$ contains exactly the graphs that have a tree-partition
of adhesion at most $k$ whose torsos are graphs in ${\cal G}$.
\end{observation}

\begin{lemma}
\label{immature}
Let $k\in\Bbb{N}_{\geq 0}$ and let $G$ be a graph and ${\cal D}=(T,{\cal B})$ be a tree-partition of $G$ of adhesion at most $k$. If $\theta_{k+1}\leq G$, then 
there is a $t\in V(T)$ such that  $\theta_{k+1}\leq Z_{t}$.
\end{lemma}

\begin{proof}
Observe that if  $\theta_{k+1}\leq G$,then there are two vertices $x$ and $y$ in $G$ 
that are connected by $k+1$ pairwise edge-disjoint $(x,y)$-paths, $P_{1},\ldots,P_{k+1}$ in $G$. As ${\cal D}$ has adhesion at most $k$, there is some $t\in V(T)$ such that $x,y\in B_{t}$.
Let $T_{1},\ldots,T_{q_t}$ be the connected components
of $T\setminus t$ and let  $z_{1},\ldots,z_{q_{t}}$ be the satelites of the $t$-torso $Z_{t}$ of ${\cal D}$.
Let $i\in[k]$ and notice that, among the edges of  the $(x,y)$-path $P_{i}$, those 
missing from  $Z_{t}$ are those that do not have endpoints in $B_{t}$.
Notice also that for every $j\in[q_{t}]$ the edges of $P_{i}$
with both endpoints in $\bigcup_{t'\in V(T_{j})}B_{t'}$ appear as consecutive edges in $P_{i}$. We now contract each such set of edges to the vertex $z_{j}$ for each $j\in[q_{t}]$
and observe that the resulting path $P_{i}'$ is a path of $Z_{t}$.
Observe that 
$P_{1}',\ldots,P_{k+1}'$ are pairwise edge-disjoint $(x,y)$-paths of $Z_{t}$ and  we conclude that 
$\theta_{k+1}\leq Z_{t}$ as required.
\end{proof}

Let ${\cal D}=(T,{\cal B})$ be a tree-partition of a graph $G$ and $k\in\Bbb{N}_{\geq 0}$. 
We say that a torso $Z_{t}$ of ${\cal D}$ is 
\begin{itemize}
\item  {\em $k$-splitable}: if it contains 
a   cut $(X,\overline{X})$ of size smaller than or equal to  $k$ where both 
$X$ and $\overline{X}$ contain  some vertex of degree at least $k+1$. 
\item {\em $k$-overloaded}: if at least two of its vertices have degree at least $k+1$.
\end{itemize}

Given a tree-partition ${\cal D}=(T,{\cal B})$, we define  $${\bf w}({\cal D})=\sum_{t\in V(T)}({\bf s}_{\cal D}(t)-1)$$ where ${\bf s}_{\cal D}(t)$ is the number of vertices in $B_{t}$ that have degree at least $k+1$. 
\begin{observation}
\label{daylight}
Let $G$ be a graph, $k\in\Bbb{N}_{\geq 0}$, and ${\cal D}$ be  a tree-partition of $G$
that  has strength at least $k+1$. Then ${\bf w}({\cal D})>0$ iff some of its torsos are $k$-overloaded.
\end{observation}

Given a $k\in\Bbb{N}_{\geq 0}$, we say that a tree-partition 
${\cal D}=({\cal B},T)$ is {\em $k$-tight} if, its adhesion is at most $k$ and 
its strength is at least $k+1$.

\begin{lemma}
\label{dispense}
For every graph $G$ and $k\in\Bbb{N}_{\geq 0}$, if ${\cal D}$ is a $k$-tight tree-partition of $G$ with a $k$-splittable torso, then there is  a  $k$-tight tree-partition ${\cal D}'$ of $G$ where ${\bf w}({\cal D}')<{\bf w}({\cal D})$.\end{lemma}

\begin{proof} Let $Z_{t}$ be a splittable torso of ${\cal D}$ and 
let  $L_t=\{z_{1},\ldots,z_{q_{t}}\}$ be the satellite vertices 
of $Z_{t}$. We denote by $t_{1},\ldots,t_{q_{t}}$ be the vertices of $T$ represented by $z_{1},\ldots,z_{q_{t}}$ respectively. Also we denote by $T_{1},\ldots,T_{q_{t}}$ the connected components of $T\setminus t$ that are subsumed by $z_{1},\ldots,z_{q_{t}}$, respectively. Let also $Q_{t}$ be the vertices of $Z_{t}$ that have degree at least $k+1$. 
As the adhesion of ${\cal D}$ is at most $k$, it follows that each vertex in $L_{t}$ has degree at most $k$. Therefore, $Q_{t}\subseteq B_{t}$.

We now construct a tree-partition ${\cal D}'$ of $G$. As $Z_{t}$ is $k$-splittable,  there is a  cut $(X,\overline{X})$ of $Z_{t}$, of size at most $k$
and two vertices $x,y$ where $\deg_{Z_{t}}(x),\deg_{Z_{t}}(y)\geq k+1$, and  $x\in X$ and $y\in\overline{X}$.
We set $Q_{t}^{(x)}=Q_{t}\cap X$ and $Q_{t}^{(y)}=Q_{t}\cap \overline{X}$
and keep in mind that  $x\in Q_{t}^{(x)}$ and $y\in Q_{t}^{(y)}$.
 Note that there is a set $I\subseteq [q_{t}]$
such that $X\cap Z_{t}=\{z_{i}\mid i\in I\}$ and  $\overline{X}\cap Z_{t}=\{z_{i}\mid i\in [q_{t}]\setminus I\}$. We  construct the tree $T'$ as follows: we start from $T\setminus t$, then  add two new adjacent vertices $t_{x}$ and $t_{y}$, make $t_{x}$ adjacent with all vertices in $\{t_{i}\mid i\in I\}$
and make $t_{y}$ adjacent with all vertices in  $\{t_{i}\mid i\in [q_{t}]\setminus I\}$.
We also define ${\cal B}'=\{B_{h}'\mid h\in V(T')\}$ such that if $h\in V(T)\setminus \{t\}$, then $B_{h}'=B_{h}$. Finally,  set $B_{t_{x}}'=B_{t}\cap  X$ and   $B_{t_{y}}'=B_{t}\cap \overline{X}$.    Observe  that
\begin{itemize} 
\item if $e=t_{x}t_{y}$, then $|{\bf cross}_{{\cal D}'}(e)|={\bf cut}_{Z_{t}}(X,\overline{X})\leq k$,  
\item if  $e= t_{y}t_{i}, i\in [q_{t}]\setminus I$, then $|{\bf cross}_{{\cal D}'}(e)|=|{\bf cross}_{\cal D}(tt_{i})|\leq k$, 
\item if $e= t_{x}t_{i}, i\in I$, then $|{\bf cross}_{{\cal D}'}(e)|=|{\bf cross}_{\cal D}(tt_{i})|\leq k$, and 
\item if $e\in E(T')\setminus E(T)|$, then $|{\bf cross}_{{\cal D}'}(e)|= |{\bf cross}_{{\cal D}}(e)|\leq k$.
\end{itemize}
From the above, we deduce that the adhesion of ${\cal D}'$ is at most $k$. 

Let now $v\in V(T')$. As ${\cal D}$ has strength  at least $k+1$, then for each $h\in V(T)\setminus \{t\}$
there is a vertex in $B'_{h}$ that has degree at least $k+1$. This, together with the fact that $x\in B_{t_{x}}'$ and  $y\in B_{t_{y}}'$ implies that ${\cal D}'$ has strength at least $k+1$. Therefore  ${\cal D}'$ is $k$-tight. \smallskip

We finally observe the following:
\begin{itemize} 
\item ${\bf s}_{{\cal D}'}(t_{x})=|Q_{t}^{(x)}|$,
\item ${\bf s}_{{\cal D}'}(t_{y})=|Q_{t}^{(y)}|$, and 
\item if $t\in V(T')\setminus \{t_x,t_y\}$, then ${\bf s}_{{\cal D}'}(t)={\bf s}_{{\cal D}}(t)$\end{itemize}
From the above, $({\bf s}_{{\cal D}'}(t_{x})-1)+({\bf s}_{{\cal D}'}(t_{x})-1)=|Q_{t}|-2=({\bf s}_{{\cal D}}(t)-1)-1$, therefore ${\bf w}({\cal D}')<{\bf w}({\cal D})$ as required.
\end{proof}

Given a tree $T$ and two members $a,a’$ of $E(T)\cup V(T)$
we define $aTa’$ as the unique path in $T$ starting from $a$ and finishing on $a’$.
Also, given a vertex $t\in V(T)$ we define its {\em status} of $t$
 as $${\bf status}(T,t)=\sum_{t'\in V(T)}|E(tTt')|,$$ i.e.,
 the sum of all the lengths of all the paths from $t$ to the rest of the vertices of $T$.

Let $(X,\overline{X})$ and $(Y,\overline{Y})$ be two cuts of a graph $G$.
We say that the cuts $(X,\overline{X})$ and $(Y,\overline{Y})$ are {\em parallel}
if $X\subseteq Y$, or  $\overline{X}\subseteq \overline{Y}$, or $X\subseteq \overline{Y}$, or  $\overline{Y}\subseteq {X}$.

\begin{lemma}\label{connects}
Let $k\in\Bbb{N}_{\geq 0}$.
If  $G$ is a $\theta_{k+1}$-immersion free  graph with at least one vertex of degree at least $k+1$,
Then $G$ has a $k$-tight tree-partition where each torso has exactly one vertex of degree greater than $k$.
\end{lemma}

\begin{proof}
Notice that $G$ has at least one $k$-tight tree-partition 
that  consists of a single bag containing all the vertices of $G$. 
Among all $k$-tight tree-partitions of $G$, consider the set $\frak{D}$ 
containing every $k$-tight tree-partition of $G$, 
where   ${\bf w}({\cal D})$ takes the {minimum} possible value, say  $\ell$. 
From \autoref{daylight} it is enough to prove that $\ell=0$, i.e.,   
the tree-partitions 
in $\frak{D}$ contain no $k$-overloaded torsos. Assume, towards a contradiction, that $\ell>0$.
Consider two vertices $x$ and $y$, of $G$ each of degree at least $k+1$, that belong to the same bag  of some 
tree-partition of $\frak{D}$. Among all tree-partitions in ${\cal D}$ containing $x,y$ in the same bag, say  $B_{t}$, we choose ${\cal D}=(T,{\cal B})$ to be   one where ${\bf status}(T,t)$ is {minimized}.

As $\theta_{k+1}\not\leq G$,  the graph $G$ contains  
some   $(x,y)$-cut $(X,\overline{X})$ of size at most $k$.
Let ${\cal S}_{x,y}$ be the set of  all such   cuts.

We say that   an edge $e\in E(T)$ is {\em crossed} by $(X,\overline{X})$ if the cut of $G$ corresponding to $\textbf{cross}_{\mathcal{D}}(e)$ and the cut $(X,\overline{X})$ are not parallel.
As both $x$ and $y$ have degree at least $k+1$, there should be two edges $e_x$ and $e_y$
in $\textbf{cross}_{\mathcal{D}}(e)$ such that $e_{x}\subseteq X$ and $e_{y}\subseteq \overline{X}$.

Let  $(X,\overline{X})\in{\cal S}_{x,y}$. 
Let $e=t’t''$ be an edge of $E(T)$ that is crossed by $(X,\overline{X})$. 
We make the convention that, whenever we consider such an edge,  
we assume that $|E(tTt'')|<|E(tTt')|$, i.e., $t''$ is closer  to $t$ than $t'$, in $T$.
We say that such an edge 
$e$ is {\em $(X,\overline{X})$-extremal for} $(T,t)$ 
if there is no other edge $e'\neq e$ of $T$ that is crossed by $(X,\overline{X})$
and such that  $e\in E(e'Tt)$.
We denote by ${\bf extr}(X,\overline{X})$ the set of edges of $T$ that are $(X,\overline{X})$-extremal for $(T,t)$. Notice that ${\bf extr}(X,\overline{X})$ should be non-empty, as otherwise 
$(X,\overline{X})$ should induce  a  cut of $Z_{t}$, therefore $Z_{t}$ would be $k$-splittable and this, due to  \autoref{dispense}, would contradict  the {minimality} of ${\bf w}({\cal D})$. 
We next  define the {\em cost} of the cut $(X,\overline{X})$ as 
%

$${\bf cost}_{T,t}(X,\overline{X})=\sum_{t't''\in {\bf extr}(X,\overline{X})} |E(tTt')|.$$

%
%
We now pick the $(x,y)$-cut $(X,\overline{X})\in {\cal S}_{x,y}$ as one of {minimum} possible cost, in other words, 
${\bf cost}(X,\overline{X})=\min\{{\bf cost}(X',\overline{X}')\mid (X',\overline{X}')\in{\cal S}_{x,y}\}$.

Let $e=t't''$ be an $(X,\overline{X})$-extremal edge of $T$.  Let  $(A,\overline{A})$  be the cut of $G$ 
whose edges are ${\bf cross}_{\cal D}(e)$ and w.l.o.g., we assume that $x,y\in A$. Recall that 
\begin{eqnarray}
\rho(X)=\rho(\overline{X})\leq k& \mbox{and
}&\rho({A})=\rho(\overline{A})\leq k .\label{handling}
\end{eqnarray}
We next claim   that 
\begin{eqnarray}
\rho(A\cap X)=|E(A\cap X, \overline{A}\cup \overline{X})|>k.\label{softened}
\end{eqnarray}
To see~\eqref{softened}, notice that if this is not the case, then 
$(A\cap X,\overline{A}\cup \overline{X})\in {\cal S}_{x,y}$, because  $x\in A\cap X$ and $y\in \overline{A}\cup \overline{X}$.
Notice that if $t=t’$, then  ${\bf extr}(A\cap X,\overline{A}\cup \overline{X})={\bf extr}(X,\overline{X})\setminus\{t''t’\}$ while if $t^*$ is the unique neighbor of $t'$ in the path joining $t'$ and $t$,
then ${\bf extr}(A\cap X,\overline{A}\cup \overline{X})={\bf extr}(X,\overline{X})\setminus\{t''t’\}\cup\{t't^*\}$. In both cases, ${\bf cost}(A\cap X,\overline{A}\cup \overline{X})={\bf cost}(X,\overline{X})-1$, a contradiction to the {minimality} of the choice of $(X,\overline{X})$.
Working symmetrically on $\overline{A}$, instead of $A$, it follows that 
\begin{eqnarray}
\rho(A\cap \overline{X})=|E(A\cap \overline{X}, \overline{A}\cup X)|>k.\label{begrudge}
\end{eqnarray}
By the submodularity of $\rho$, we have that
\begin{eqnarray}
\rho(A\cap X)+\rho( A\cup X)\leq \rho(X)+\rho(A).
\label{bringing}\\
\rho({A}\cap \overline{X})+\rho( {A}\cup \overline{X})\leq \rho(\overline{X})+\rho({A}).
\label{operetta}
\end{eqnarray}
\begin{figure}[h]
	\begin{center}
			{\includegraphics[width=120mm]{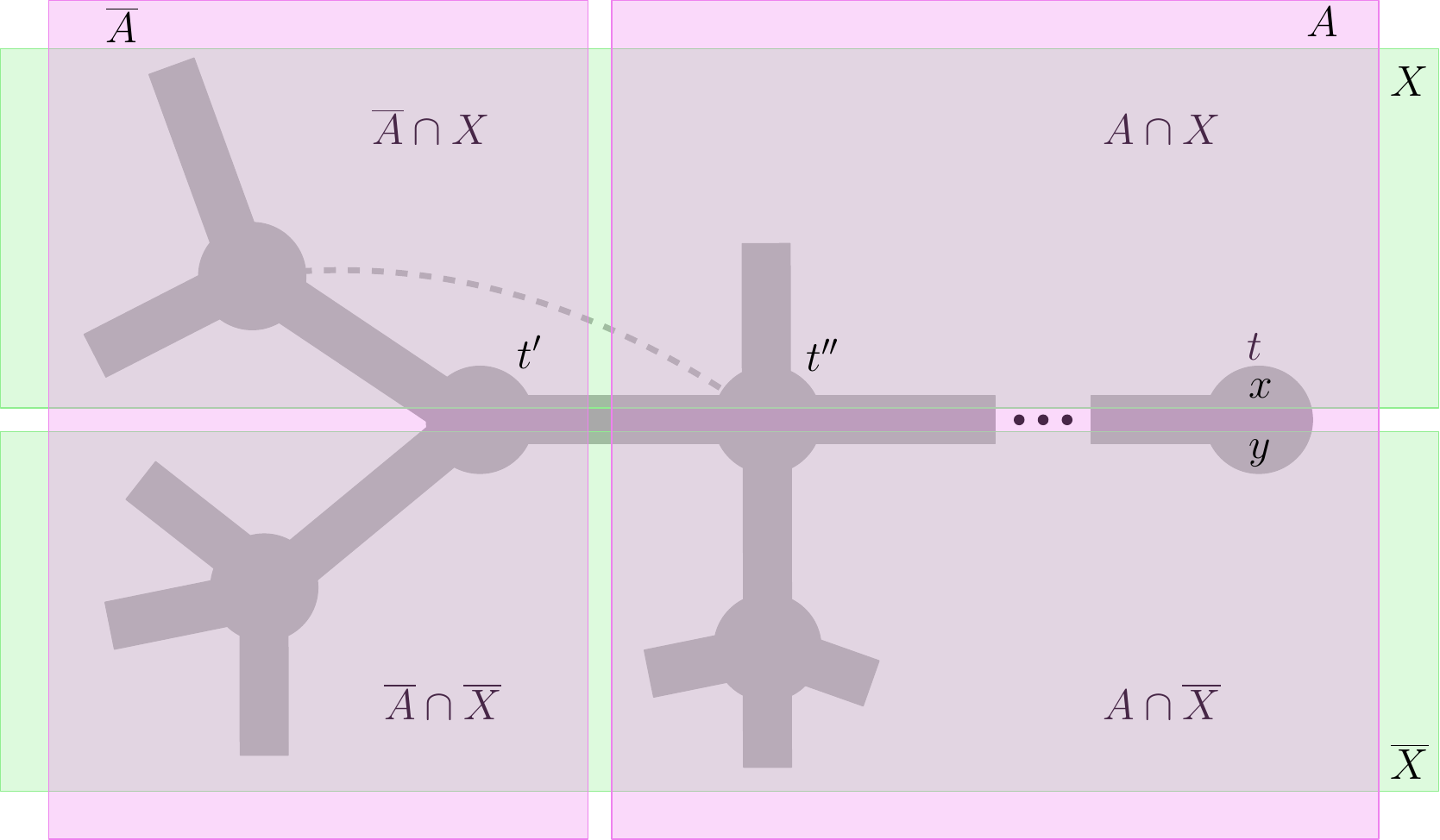}}
	\end{center}
	\caption{A visualization of the proof of \autoref{connects}.}
	\label{myidx}
	\end{figure}
Combining now \eqref{handling}, \eqref{softened}, and~\eqref{bringing}
and \eqref{handling}, \eqref{begrudge}, and~\eqref{operetta} we have that $\rho( A\cup X)\leq k  \mbox{~and~}  \rho(A\cup \overline{X})\leq k$ which can be rewritten
\begin{eqnarray}
\rho(\overline{A}\cap \overline{X})\leq k & \mbox{~and~} & \rho(\overline{A}\cap X)\leq k .\label{uttering}
\end{eqnarray}
Note that the vertices of $B_{t'}$ that have degree at least $k+1$
should all be in exactly one of $\overline{A}\cap \overline{X}$ and $\overline{A}\cap X$.
Indeed, if this is not correct, then $Z_{t'}$ should be $k$-splittable  and this, due to  \autoref{dispense}, would contradict  the {minimality} of ${\bf w}({\cal D})$. W.l.o.g. we assume that  $Q=B_{t'}\cap\overline{A}\cap{X}$ contains 
only vertices of degree at most $k$. 

Let $z_{1},\ldots,z_{q_{t'}}$ be the satellites of $Z_{t'}$
and let $t_{i}$ be the vertex of $T$ represented by $z_{i},i\in[q_{t'}]$,  
assuming, w.l.o.g., that $z_{1}$ represents $t''$ in $T$ (that is $t_{1}=t''$). Let also $T_{i}$
be the connected component of $T\setminus t’$  subsumed by $z_{i}$, for $i\in[q_{t'}]$.
As $t't''\in {\bf extr}(X,\overline{X})$, there is some  non-empty  $I\subseteq [2,q_{t'}]$ such that 
\begin{eqnarray}
\bigcup_{i\in I}\bigcup_{s\in V(T_{i})}B_{s}= (\overline{A}\cap {X})\setminus B_{t’}\mbox{~and~}\bigcup_{i\in [2,q_{t'}]\setminus I}\bigcup_{s\in V(T_{i})}B_{s}= ( \overline{A}\cap \overline{X})\setminus B_{t'}.
\end{eqnarray}

We now add the set $Q$ to $B_{t''}$ and remove it from $B_{t'}$, and also  remove from 
$T$ all edges in $\{t_{i}t'\mid i\in I\}$ and add the edges $\{t_{i}t''\mid i\in I\}$ to get $T'$ (in \autoref{myidx}, the new edge is  depicted by the dashed edge).
Observe that ${\cal D}'=(T',{\cal B})$ is a tree-partition of $G$ with adhesion at most $k$
and where all its nodes contain some vertex of degree at least $k+1$. 
Therefore ${\cal D}'$ is $k$-tight. 
Notice that, by the construction of $T'$,   ${\bf status}(T',t)<{\bf status}(T,t)$ a contradiction to the {minimality}
of ${\bf status}(T,t)$ in the choice of ${\cal D}=(T,{\cal B})$.\end{proof}

\begin{theorem}
\label{furrowed}
For every graph $G$ and $k\in \Bbb{N}$, $G$ is $\theta_{k+1}$-immersion free if and only if  $G\in {\cal A}_{k}^{(\leq k)}$.
\end{theorem}

\begin{proof}
We prove first  ``only if'' direction.  If $G$ has no vertices  of degree at least $k+1$, then $G\in {\cal A}_{k}$ and the result follows trivially.
If $G$ has at least one vertex of degree at least $k+1$, then, because of \autoref{connects},  $G$ has a $k$-tight tree-partition of adhesion at most $k$ and whose torsos belong to ${\cal A}_{k}$. Then, from \autoref{teaching}, $G\in {\cal A}_{k}^{(\leq k)}$.

We next prove the ``if'' direction. Suppose that $G\in {\cal A}_{k}^{(\leq k)}$, therefore, from  \autoref{teaching}, $G$ has a  tree-partition ${\cal D}$  of adhesion at most $k$ whose torsos are all in ${\cal A}_{k}$.
As none of the torsos of ${\cal D}$ contains $\theta_{k+1}$ as an immersion,
because of \autoref{immature}, the same holds for $G$ and we are done.
\end{proof}

As mentioned by one of the reviewers, \autoref{furrowed} can alternatively  be proved by a suitable 
application of the theorem of Gomory and Hu~\cite{GomoryH61mult} (see also~\cite{DiestelHL19prof} and~\cite{DeVosMMS13anot}).

%

\subsection{An upper bound to edge-admissibility}

In this subsection we prove that $\theta_{k+1}$-immersion free graphs have edge-admissibility at most $2k-1$. In the end of this section, this will serve for proving  \autoref{presents}.

\paragraph{Carving decompositions.}
Given a tree $T$ we denote by $L(T)$ the set of all the vertices of $T$ that have degree at most $1$ and we call them the {\em leaves}  of $T$.
A {\em rooted tree} is a pair ${\bf T}=(T,r)$ where $T$ is a tree and $r\in V(T)$.
A {\em binary rooted tree} is a rooted tree ${\bf T}=(T,r)$ where all its non-leaf vertices have exactly two children.
If $v\in V(T)$, we define  
${\bf descl}_{{\bf T}}(v)$ as the set containing every leave $\ell$ of $T$ such that $v\in V(rT\ell)$.

Let $G$ be a graph and $S\subseteq V(G)$.
A {\em rooted carving decomposition} of  $G$ is a pair  $({\bf T},\sigma)$
consisting of  a rooted binary tree  ${\bf T}=(T,r)$ and a function  $\sigma:V(G)\to L(T)$.
We stress that $\sigma$ is not a bijection, i.e., we permit many vertices of $G$ to be mapped 
to the same leaf of $T$.
The {\em weight} of a vertex $t$ in $V(T)\setminus L(T)$ is defined as  
$${\bf w}(t)=|E_{G}(S_{1},S_{2})|$$ 
where $S_{i}=\sigma^{-1}({\bf descl}_{{\bf T}}(t_{i})),i\in[2]$
and 
$t_{1},t_{2}$ are the children of $t$ in $T$.
For every edge $e=tt'$ of $E(T)$, where $t'$ is a child of $t$,  we  define ${\bf cut}(e)$ as the set  $E_{G}(V_{1},V_{2})$
where $V_{1}=\sigma^{-1}({\bf descl}_{\bf T}(t’))$ and $V_{2}=V(G)\setminus V_{1}$.
We also define the {\em weight} of $e=tt’$ as ${\bf w}(e)=|{\bf cut}(e)|$.

\begin{lemma}
\label{holiness}
Let $G$ be a graph and $k\in\Bbb{N}_{\geq 1}$. If  $\theta_{k+1}\nleq G$, then $\eadm(G)\leq 2k-1$.
\end{lemma}
%

\begin{proof}
We show  that if $G$ is $\theta_{k+1}$-immersion free, then  $G$ cannot contain 
a $(2k,\infty)$-edge-hideout and therefore, from \autoref{consiste}, $\eadm(G)\leq  2k-1$.
Suppose to the contrary that  $S, |S|\geq 2$, is  a $(2k,\infty)$-edge-hideout of $G$. We build a rooted carving decomposition of  $G$ by applying the following precedure:\medskip

\noindent {\bf Step 1}. Consider $({\bf T},\sigma)$ where ${\bf T}=(T,v)$,
$T$ consists of only one vertex, that is the root $r$, and  
 $\sigma(v)=r$ for all $v\in V(G)$.
 
\noindent {\bf Step 2}. Let $\ell$ be a vertex of $T$ where $|\sigma^{-1}(\ell)\cap S|\geq 2$. If no such vertex exists, then {\bf stop}.

\noindent {\bf Step 3}.  Pick, arbitrarily, two distinct vertices $x_1$ and $x_2$ in $\sigma^{-1}(\ell)\cap S$. Notice that $G$ contains a $(x_1,x_2)$-cut $(X^{1},X^{2})$ of at most $k$ edges where $x_{i}\in X^{i}, i\in[2]$, otherwise, from Menger's theorem
there are $k+1$ pairwise edge disjoint paths from $x_{1}$ to $x_{2}$ in $G$, which implies the existence 
of $\theta_{k+1}$ as an immersion in $G$, a contradiction.
We now add in $T$ two new vertices $\ell_{1}$ and $\ell_{2}$ 
make them the children of $\ell$ and update $\sigma$ so that the vertices in $X^{i}\cap \sigma^{-1}(\ell)$ are now mapped in $\ell_{i}, i\in [2]$, i.e. we remove from $\sigma$ $(t,\sigma^{-1}(\ell))$ and we add $(t_{1},X^{1}\cap \sigma^{-1}(\ell))$ and $(t_{2},X^{2}\cap \sigma^{-1}(\ell))$.

\noindent {\bf Step 4}. Go to {\bf Step 2}.\medskip

Let  $({\bf T},\sigma)$ be the rooted carving decomposition produced by the above procedure.
By the construction of  $({\bf T},\sigma)$, each vertex of $T$ has weight  at most $k$ and for each leaf $\ell\in L(G)$, $|\sigma^{-1}(\ell)\cap S|=1$.
We construct a path $P$ of $T$  by applying the following procedure.\medskip

\noindent {\bf Step 1}. Let $P$ be the path of $T$ consisting 
of $r$ and one (arbitrarily chosen), say $t’$, of the children of $r$ (i.e., $P$ is just an edge).
Notice that ${\bf w}(\{r,t’\})={\bf w}(r)\leq k\leq 2k-1$ (recall that $k\geq 1$).

\noindent {\bf Step 2}. Let $e$ be the the last edge of $P$ (starting from $r$) 
and let $t$ be its endpoint that 
is also an endpoint of $P$ (different than $r$).   If $t$ is a leaf of $T$, then {\bf stop}.

\noindent {\bf Step 3}. Let $t_{1}$ and $t_{2}$ be the children of $t$ and let $e_{i}=tt_i, i\in[2]$.
 We partition the edges of ${\bf  cut}(e)$  into two sets, namely  $F_{1}$ and $F_{2}$ so that $F_{i}$
contains edges with an endpoint in ${\bf descl}_{{\bf T}}(t_{i}), i\in[2]$. Notice that 
${\bf  cut}(e_i)  =  F_{i}\cup E_{G}(\sigma^{-1}({\bf descl}_{{\bf T}}(t_{1})),\sigma^{-1}({\bf descl}_{{\bf T}}(t_{2}))),$ therefore, for $i \in[2]$,
\begin{eqnarray}
{\bf w}(e_i) =  |{\bf  cut}(e_i)| & = & |F_{i}|+|{\bf w}(t)|.\label{supports}
\end{eqnarray}
As ${\bf w}(e)\leq 2k-1$, one, say $F_{1}$, of $F_{1}, F_{2}$ should have at most $k-1$ edges. By applying~\eqref{supports} for $i=1$, we obtain that $|{\bf w}(e_1)|\leq k-1+{\bf w}(t)\leq 2k-1$. We now 
extend $P$ by adding in it the vertex $t_{1}$ and the edge $e_1$ and we update $e:=e_1$.

\noindent {\bf Step 4}. Go to {\bf Step 2}.\medskip

We just constructed a path $P$ in $T$ between $r$ and a leaf of $\ell$ of $T$ 
such that for every edge $e\in E(P)$,  ${\bf w}(e)\leq 2k-1$.
Notice that $\sigma^{-1}(\ell)$ contains exactly one vertex, say $x$, of $S$.
Moreover, if $f$ is the edge of $T$ that is incident to $\ell$,
then $\rho(\sigma^{-1}(\ell))={\bf w}(f)\leq 2k-1$, as $f$ is an edge of $P$ (the last one).
This implies that there is a set of $2k-1$ edges blocking every path from $x$ to $S\setminus \{x\}$.
Therefore, ${\bf supp}_G(\infty,x,S\setminus\{
x\})\geq 2k-1$, contradicting to the fact that 
$S$ is  a $(2k,\infty)$-edge-hideout of $G$.
\end{proof}

\begin{observation}
\label{conceive}
If $H$ and $G$ are graphs then $H\leq G\Rightarrow \eadm(H)\leq \eadm(G)$. 
\end{observation}

\begin{proof}
Suppose that $H\leq G$ and that $k\leq \eadm(H)$.
From \autoref{consiste} $H$ contains a $(k,\infty)$-edge-hide-out $S\subseteq V(H)$.
Because of  Menger's theorem, for every vertex $v\in S$
there are at least $k+1$ pairwise edge-disjoint paths from $v$ to vertices of $S\setminus \{v\}$.
Notice that these paths also exist in $G$ as  the ``inverse'' of the  lift operation does not 
alter the paths from a vertex of $S$ to the rest of the vertices of $S$. These paths,  again using 
 Menger's theorem, imply that $S$ is also a $(k+1,\infty)$-edge-hide-out of $G$, therefore, again from \autoref{consiste}, $k\leq \eadm(G)$.
\end{proof}

We are now ready to give the proof of \autoref{presents}.

\begin{proof}[Proof of \autoref{presents}]
For the first part of the theorem, observe that $\eadm(\theta_{k+1})=k+1$, therefore, from \autoref{conceive}, $\theta_{k+1}\nleq G$. Using now  the ``only if'' direction of \autoref{furrowed} we obtain that $G\in {\cal A}_{k}^{(\leq k)}$, as required.

For the second part of the theorem, let $G\in {\cal A}_{k}^{(\leq k)}$, which by the ``if'' direction of \autoref{furrowed} implies that $\theta_{k+1}\nleq G$. Using now \autoref{holiness}, we conclude that  $\eadm(G)\leq 2k-1$.
\end{proof}

\section*{Acknowledgements} We wish to thank the anonymous reviewers for their comments and remarks that considerably improved the presentation of this paper.

\bibliographystyle{plain}
 \newcommand{\bibremark}[1]{\marginpar{\tiny\bf#1}}
  \newcommand{\biburl}[1]{\url{#1}}

 \end{document}